\RequirePackage{fix-cm}
\documentclass{article}          
\usepackage{authblk}
\usepackage{fullpage}
\usepackage{amsthm}
\usepackage{graphicx}
\usepackage{url}
\usepackage{xcolor}
\usepackage{natbib}
\usepackage{amsmath}
\usepackage{amssymb}
\usepackage{algorithm}
\usepackage{algpseudocode}
\newcommand{\mb}{\mathbf}

\newcommand{\E}{\mathbb{E}}
\newcommand{\V}{\mathbb{V}\mathrm{ar}}
\newcommand{\C}{\mathbb{C}\mathrm{ov}}

\graphicspath{{./img/}}

\newtheorem{ass}{Assumption}
\newtheorem{thm}{Theorem}
\newtheorem{prop}{Proposition}
\newtheorem{cor}{Corollary}
\newtheorem{lem}{Lemma}

\newcommand{\norm}[1]{\left\lVert #1 \right\rVert}
\begin{document}

\title{Control Variates for Stochastic Gradient MCMC
}


\author{Jack Baker$^1$\footnote{email: \texttt{j.baker1@lancaster.ac.uk}}\hskip 2em 
        Paul Fearnhead$^1$\hskip 2em 
        Emily B. Fox$^2$\hskip 2em 
        Christopher Nemeth$^1$}
\affil{$^1$ STOR-i Centre for Doctoral Training, Department of Mathematics and Statistics, Lancaster University, Lancaster, UK \\ $^2$ Department of Statistics, University of Washington, Seattle, WA}

\date{}

\maketitle

\begin{abstract}
It is well known that Markov chain Monte Carlo (MCMC) methods scale poorly with dataset size. A popular class of methods for solving this issue is stochastic gradient MCMC. These methods use a noisy estimate of the gradient of the log posterior, which reduces the per iteration computational cost of the algorithm. Despite this, there are a number of results suggesting that stochastic gradient Langevin dynamics (SGLD), probably the most popular of these methods, still has computational cost proportional to the dataset size. We suggest an alternative log posterior gradient estimate for stochastic gradient MCMC, which uses control variates to reduce the variance. We analyse SGLD using this gradient estimate, and show that, under log-concavity assumptions on the target distribution, the computational cost required for a given level of accuracy is independent of the dataset size. Next we show that a different control variate technique, known as zero variance control variates can be applied to SGMCMC algorithms for free. This post-processing step improves the inference of the algorithm by reducing the variance of the MCMC output. Zero variance control variates rely on the gradient of the log posterior; we explore how the variance reduction is affected by replacing this with the noisy gradient estimate calculated by SGMCMC.

\end{abstract}

\noindent\textbf{Keywords:} Stochastic gradient MCMC; Langevin dynamics; scalable MCMC; control variates; computational cost; big data

\section{Introduction}
\label{sec:intro}

Markov chain Monte Carlo (MCMC), one of the most popular methods for Bayesian inference, scales poorly with dataset size. This is because standard methods require the whole dataset to be evaluated at each iteration. Stochastic gradient MCMC (SGMCMC) are a class of MCMC algorithms that aim to account for this issue. The algorithms have recently gained popularity in the machine learning literature, though they were originally proposed in \cite{Lamberton2002}. These methods use efficient MCMC proposals based on discretised dynamics that use gradients of the log posterior. They reduce the computational cost by replacing the gradients with an unbiased estimate which uses only a subset of the data, referred to as a minibatch. They also bypass the acceptance step by making small discretisation steps \citep{{Welling2011, Chen2014, Ding2014, Dubey2016}}. 

These new algorithms have been successfully applied to a range of state of the art machine learning problems \citep[e.g.][]{Patterson2013, Li2016}. There is a variety of software available implementing these methods \citep{{Tran2016, Baker2017}}. In particular \cite{Baker2017} implements the control variate methodology we discuss in this article.

This paper investigates stochastic gradient Langevin dynamics (SGLD), a popular SGMCMC algorithm that discretises the Langevin diffusion. There are a number of results suggesting that while SGLD has a lower per iteration computational cost compared with MCMC, its overall computational cost is proportional to the dataset size \citep{{Welling2011, Nagapetyan2017}}. This motivates improving the computational cost of SGLD, which can be done by using control variates \citep{Ripley2009}. Control variates can be applied to reduce the Monte Carlo variance of the gradient estimate in stochastic gradient MCMC algorithms. We refer to SGLD using this new control variate gradient estimate as SGLD-CV.

We analyse the algorithm using the Wasserstein distance between the distribution defined by SGLD-CV and the true posterior distribution, by adapting recent results by \cite{Dalalyan2017}. Our results are based on assuming the target distribution is strongly log-concave. We get bounds on the Wasserstein distance between the target distribution and the distribution we sample from at a given step of SGLD-CV. These bounds are in terms of the tuning constants chosen when implementing SGLD-CV. By making assumptions on how the posterior distribution changes with the number of data points, we are able to show that the computational cost required for a given level of accuracy does not grow with the dataset size. Though this is providing we can obtain an estimate, $\hat \theta$, of the posterior mode that is sufficiently close. Our results also show the impact on the computational cost of using a poor $\hat{\theta}$ value.

The algorithm requires some additional preprocessing steps before the computational cost benefits come into effect. These preprocessing steps include finding $\hat \theta$ and calculating the full log posterior gradient at $\hat \theta$. Both of these steps have a cost that is linear in the amount of data. However, the cost of finding $\hat \theta$ essentially replaces the burn-in of the chain, and we find empirically that this is often more efficient. The cost of these steps is analysed in more detail in the article.

The use of control variates has also been shown to be important for other Monte Carlo algorithms for simulating from a posterior with a cost that is sub-linear in the number of data points \citep{{Bardenet2016, Bierkens2016, Pollock2016, Nagapetyan2017}}. For previous work that suggests using control variates within SGLD see \cite{Dubey2016} and  \cite{Chen2017}. These latter papers, whilst showing benefits of using control variates, do not show that the resulting algorithm can have sub-linear cost in the number of data points. A recent paper, \cite{Nagapetyan2017}, does investigate how SGLD-CV peforms in the limit as you have more data under similar log-concavity assumptions on the posterior distribution. They have results that are qualitatively similar to ours, including the sub-linear computational cost of SGLD-CV. Though they measure accuracy of the algorithm through the mean square error of Monte Carlo averages rather than through the Wasserstein distance.

Not only can control variates be used to speed up stochastic gradient MCMC by enabling smaller minibatches to be used; we show that they can be used to improve the inferences made from the MCMC output. In particular, we can use post-processing control variates \citep{{Mira2013,Papamarkou2014,Friel2016}} to produce MCMC samples with a reduced variance. The post-processing methods rely on the MCMC output as well as gradient information. Since stochastic gradient MCMC methods already compute estimates of the gradient, we explore replacing the true gradient in the post-processing step with these free estimates. We also show theoretically how this affects the variance reduction factor; and empirically demonstrate the variance reduction that can be achieved from using these post-processing methods.

\section{Stochastic gradient MCMC}
\label{sec:review}

Throughout this paper we aim to make inference on a vector of parameters $\theta \in \mathbb R^d$, with data $\mb x = \{x_i\}_{i=1}^N$. We denote the probability density of $x_i$ as $p(x_i | \theta)$ and assign a prior density $p(\theta)$. The resulting posterior density is then $p( \theta | \mb x ) \propto p( \theta ) \prod_{i=1}^N p( x_i | \theta )$, which defines the posterior distribution $\pi$. For brevity we write $f_i(\theta) = - \log p(x_i | \theta)$ for $i = 1, \dots N$, $f_0(\theta) = - \log p(\theta)$ and $f(\theta) = - \log p(\theta | \mb x)$.

Many MCMC algorithms are based upon discrete-time approximations to continuous-time dynamics, such as the Langevin diffusion, that are known to have the posterior as their invariant distribution. The approximate discrete-time dynamics are then used as a proposal distribution within a Metropolis-Hastings algorithm. The accept-reject step within such an algorithm corrects for any errors in the discrete-time dynamics. Examples of such an approach include the Metropolis-adjusted Langevin algorithm (MALA; see e.g. \cite{Roberts1998}) and Hamiltonian Monte Carlo (HMC; see \cite{Neal2010}).

\subsection{Stochastic gradient Langevin dynamics}
\label{sec:sgld-review}

SGLD, first introduced by \cite{Lamberton2002}, and popularised more recently by \cite{Welling2011}, is a minibatch version of the Metropolis-adjusted Langevin algorithm. At each iteration it creates an approximation of the true gradient of the log-posterior by using a small sample of data.

{
The SGLD algorithm is based upon the discretisation of a stochastic differential equation known as the Langevin diffusion. A Langevin diffusion for a parameter vector $\theta$ with posterior $p(\theta | \mb x) \propto \exp(-f(\theta))$ is given by 
\begin{equation}
  \label{eq:langevin}
    \theta_t =  \theta_0 - \int_0^t \nabla f(\theta_s) ds + \sqrt 2 dB_t,  
\end{equation}
where $B_t$ is a $d$-dimensional Wiener process. The stationary distribution of this diffusion is $\pi$. This means that it will target the posterior exactly, but in practice we need to discretize the dynamics to simulate from it, which introduces error. A bottleneck for this simulation is that calculating $\nabla f( \theta)$ is an $O(N)$ operation. So to get around this, \cite{Welling2011} replace the log posterior gradient with the following unbiased estimate
\begin{equation}
    \nabla \hat f(\theta) := \nabla f_0(\theta) + \frac{N}{n} \sum_{i \in S_k} \nabla f_i(\theta)
    \label{eq:std-grad}
\end{equation}
for some subsample $S_k$ of $\{1,\dots,N\}$, with $|S_k| = n$. A single update of SGLD is then
\begin{equation}
    \theta_{k+1} = \theta_k - \frac{h_k}{2} \nabla \hat f(\theta_k) + \zeta_k,
\end{equation}
where $\zeta_k \sim N(0,h_k)$.
}

MALA uses a Metropolis-Hastings accept-reject step to correct for the discretisation of the Langevin process. \citet{Welling2011} bypass this acceptance step, as it requires calculating $p( \theta | \mb x )$ using the full dataset, and instead use an adaptive rather than fixed stepsize, {where $h_k \rightarrow 0$ as $k \rightarrow \infty$}. The motivation is that the noise in the gradient estimate disappears faster than the process noise, so eventually, the algorithm will sample the posterior approximately. In practice, we found the algorithm does not mix well when the stepsize is decreased to zero, so in general a fixed small stepsize $h$ is used in practice, as suggested by \citet{Vollmer2015}.

\section{Control variates for SGLD efficiency}
\label{sec:eff}

The SGLD algorithm has a reduced per iteration computational cost compared to traditional MCMC algorithms. However, there have been a number of results suggesting that the overall computational cost of SGLD is still $O(N)$ \citep{{Welling2011, Nagapetyan2017}}. The main reason for this result is that in order to control the variance in the gradient estimate, $\nabla \hat f(\theta)$, we need $n$ to increase linearly with $N$. So intuition would suggest trying to reduce this variance would reduce the computational cost of the algorithm.

A natural choice is to reduce this variance through control variates \citep{Ripley2009}. Control variates applied to SGLD have also been investigated by \citet{Dubey2016} and \cite{Chen2017}, who show that the convergence bound of SGLD is reduced when they are used. Theoretical results, similar to ours below, on how the use of control variates can improve how the computational cost of SGLD scales with $N$ are given in  \cite{Nagapetyan2017}. 

In Section \ref{sec:contr-vari-grad}, we show how control variates can be used to reduce the variance in the gradient estimate in SGLD, leading to the algorithm SGLD-CV. Then in Section \ref{sec:comp-cost-sgld} we analyse the Wasserstein distance between the distribution defined by SGLD-CV and the true posterior. There are a number of quantities that affect the performance of SGLD-CV, including the stepsize $h$, the number of iterations $K$ and the minibatch size $n$. We provide sufficient conditions on $h$, $K$ and $n$ in order to bound the Wasserstein distance. We show under certain assumptions, the computational cost, measured as $Kn$, required to bound the Wasserstein distance is independent of $N$.

\subsection{Control variates for SGMCMC}
\label{sec:contr-vari-grad}

Let $\hat\theta$ be a fixed value of the parameter, chosen to be close to the mode of the posterior $p( \theta | \mb x )$. The log posterior gradient can then be re-written as
\[
\nabla f(\theta) = \nabla f( \hat\theta ) + [\nabla f(\theta) - \nabla f(\hat \theta) ],
\]
where the first term on the right-hand side is a constant and the bracketed term on the right-hand side can be unbiasedly estimated by
\[
    \left[\nabla \hat f(\theta) - \nabla \hat f(\hat \theta)\right] = \nabla f_0(\theta) - \nabla f_0(\hat \theta) +  \frac{1}{n} \sum_{i \in S}\frac{1}{p_i} \left[ \nabla f_i(\theta) - \nabla f_i(\hat \theta) \right]
\]
where $p_1,\ldots,p_N$ are user-chosen, strictly positive probabilities, $S$ is a random sample from $\{1,\ldots,N\}$ such that $|S| = n$ and the expected number of times $i$ is sampled is $np_i$. The standard implementation of control variates would set $p_i=1/N$ for all $i$. Yet we show below that there can be advantages in having these probabilities vary with $i$; for example to give higher probabilities to sampling data points for which $\nabla f_i(\theta) - \nabla f_i(\hat \theta)$ has higher variability.

If the gradient of the likelihood for a single observation is smooth in $\theta$ then we will have
\[
    \nabla f_i(\theta) \approx \nabla f_i(\hat \theta) \ \ \mbox{if} \ \  \theta\approx\hat\theta.
\]
Hence for $\theta\approx\hat\theta$ we would expect the unbiased estimator
\begin{equation}
    \label{eq:eff-grad}
    \nabla \tilde f(\theta) = \nabla f( \hat \theta ) + [\nabla \hat f(\theta) - \nabla \hat f(\hat \theta) ],
\end{equation}
to have a lower variance than the simpler unbiased estimator \eqref{eq:std-grad}. This is because when $\theta$ is close to $\hat \theta$ we would expect the terms $\nabla \hat f(\theta)$ and $\nabla \hat f(\hat \theta)$ to be correlated. This reduction in variance is shown formally in Lemma 1, stated in Section \ref{sec:variance-reduction}.

The gradient estimate \eqref{eq:eff-grad} can be substituted into any stochastic gradient MCMC algorithm in place of $\nabla \hat f(\theta)$. We refer to SGLD using this alternative gradient estimate as SGLD-CV. The full procedure is outlined in Algorithm \ref{alg:SGLD-CV}.

\begin{algorithm}[t]
    \caption{SGLD-CV}
    \begin{algorithmic}
        \Require $\hat \theta$, $\nabla f( \hat \theta)$, $\epsilon$. \\
        Set $\theta_0 \leftarrow \hat \theta$.
        \For{$t \in 1, \dots, T$}
            \State Update $\nabla \tilde f( \theta_k)$ using \eqref{eq:eff-grad}
            \State Draw $\zeta_t \sim N( 0, \epsilon I )$
            \State $\theta_{k+1} \leftarrow \theta_k - \frac{h}{2} \nabla \tilde f(\theta_k) + \zeta_k$
        \EndFor
    \end{algorithmic}
    \label{alg:SGLD-CV}
\end{algorithm}

Implementing this in practice means finding a suitable $\hat\theta$, which we refer to as the \emph{centering value}. We show below that for the computational cost of SGLD-CV to be $O(1)$, we require  both $\hat \theta$ and the starting point of SGLD-CV, $\theta_0$, to be a $O(N^{-\frac 1 2})$ distance from the posterior mean.


In practice, we find $\hat\theta$ using \emph{stochastic optimisation} \citep{Robbins1951}, and then calculate the full log posterior gradient at this point $\nabla f(\hat \theta)$. We then start the algorithm from $\hat \theta$. In our implementations we use a simple stochastic optimisation method, known as \emph{stochastic gradient descent} \citep[SGD, see e.g.][]{Bottou2010}. The method works similarly to the standard optimisation method gradient descent, but at each iteration replaces the true gradient of the function with an unbiased estimate. A single update of the algorithm is as follows
\begin{equation}
    \theta_{k+1} = \theta_k - h_k \nabla \hat f(\theta),
\end{equation}
where $\nabla \hat f(\theta)$ is as defined in \eqref{eq:std-grad} and $h_k > 0$ is a small tuning constant referred to as the stepsize. Provided the stepsizes $h_k$ satisfy the following conditions $\sum_k h_k^2 < \infty$ and $\sum_k h_k = \infty$ then this algorithm will converge to a local maximum.

We show in Section \ref{sec:setup}, under our assumptions of log-concavity of the posterior, that finding $\hat\theta$ using SGD has a computational cost that is linear in $N$, and we can achieve the required accuracy with just a single pass through the data. As we then start SGLD-CV with this value for $\theta$, we can view finding the centering value as a replacement for the burn-in phase of the algorithm, and we find, in practice, that the time to find a good $\hat \theta$ is often quicker than the time it takes for SGLD to burn-in. One downside of this procedure is that the SGD algorithm, as well as the SGLD-CV algorithm itself needs to be tuned, which adds to the tuning burden.

In comparison to SGLD-CV, the SAGA algorithm by \cite{Dubey2016} also uses control variates to reduce the variance in the gradient estimate of SGLD. They show that this reduces the MSE of SGLD. The main difference is that their algorithm uses a previous state in the chain as the control variate, rather than an estimate of the mode. This means that SAGA does not require the additional optimisation step, so tuning should be easier. However we show in the experiments of Section \ref{sec:experiments}, that the algorithm gets more easily stuck in local stationary points, especially during burn-in. For more complex examples, the algorithm was prohibitively slow to burn-in because of this tendency to get stuck. \cite{Dubey2016} also do not show that SAGA has favourable computational cost results. 

\subsection{Variance reduction}
\label{sec:variance-reduction}

The improvements of using the control variate gradient estimate \eqref{eq:eff-grad} over the standard \eqref{eq:std-grad} become apparent when we calculate the variances of each. For our analysis, we make the assumption that the posterior is strongly log-concave, formally defined in Assumption \ref{ass:stronvex}. This has become a common assumption when analysing gradient based samplers that do not have an acceptance step \citep{{Durmus2016, Dalalyan2017}}. In all the following analysis we use $\norm{\cdot}$ to denote the Euclidean norm.

\begin{ass}
    \label{ass:stronvex}
     Strongly log-concave posterior: there exists positive constants $m$ and $M$, such that the following conditions hold for the negative log posterior
    \begin{align}
        f(\theta) - f(\theta') - \nabla f(\theta')^\top (\theta - \theta') &\geq \frac m 2 \norm{\theta - \theta'}^2 \\
        \norm{\nabla f(\theta) - \nabla f(\theta')} &\leq M \norm{\theta - \theta'}.
    \end{align}
    for all $\theta, \theta' \in \mathbb R^d$.
\end{ass}

We further need a Lipschitz condition for each of the likelihood terms in order to bound the variance of our control-variate estimator of the gradient.
\begin{ass}
 \label{ass:Lipschitz}
 Lipschitz: there exists constants $L_0,\ldots,L_N$ such that
 \[
         \norm{\nabla f_i(\theta) - \nabla f_i(\theta')} \leq L_i \norm{\theta - \theta'}, \mbox{ for $i=0,\ldots,N$}.
 \]
\end{ass}

Using Assumption \ref{ass:Lipschitz} we are able to derive a bound on the variance of the gradient estimate of SGLD-CV. This bound is formally stated in Lemma \ref{lem:grad-noise}.

\begin{lem}
    \label{lem:grad-noise}
    Under Assumption \ref{ass:Lipschitz}. Let $\theta_k$ be the state of SGLD-CV at the $k^{th}$ iteration, with stepsize $h$ and centering value $\hat \theta$.
    Assume we estimate the gradient using the control variate estimator with $p_i=L_i/\sum_{j=1}^N L_j$ for $i=1,\ldots,N$. Define $\xi_k := \nabla \tilde f(\theta_k) - \nabla f(\theta_k)$, so that $\xi_k$ measures the noise in the gradient estimate $\nabla \tilde f$ and has mean 0. Then for all $\theta_k, \hat \theta \in \mathbb R^d$, for all $k = 1, \dots, K$ we have
    \begin{equation} 
        \E\norm{\xi_k}^2 \leq \frac{\left(\sum_{i=1}^N L_i\right)^2}{n} \E\norm{\theta_k - \hat \theta}^2.
        \label{grad-var}
    \end{equation}
\end{lem}
All proofs are relegated to the Appendix. 
If Assumption \ref{ass:stronvex} also holds, then we can choose an $M=\sum_{i=0}^N L_i$, and our bound \eqref{grad-var} implies 
\[
      \E\norm{\xi_k}^2 \leq \frac{M^2}{n} \E\norm{\theta_k - \hat \theta}^2.
\]
We will use this form of the bound for the rest of the analysis. While it looks like picking $p_i$ will require estimates of the Lipschitz constants $L_i$; in practice, under Assumption \ref{ass:linear} stated below, we can just use the standard $p_i = 1 / N$ for all $i$. We use $p_i = 1 / N$ in all our implementations in the experiments of Section \ref{sec:experiments}.

In order to consider how SGLD-CV scales with $N$ we need to make assumptions on the properties of the posterior and how these change with $N$. To make discussions concrete we will focus on the following, strong, assumption that each likelihood-term in the posterior is strongly log-concave. As we discuss later, our results apply under weaker conditions.

\begin{ass}
    \label{ass:linear}
       Assume there exists positive constants $L$ and $l$ such that $f_i$ satisfies the following conditions
    \begin{align*}
        f_i(\theta) - f_i(\theta') - \nabla f_i(\theta')^\top (\theta - \theta') &\geq \frac l 2 \norm{\theta - \theta'}^2 \\
        \norm{\nabla f_i(\theta) - \nabla f_i(\theta')} &\leq L \norm{\theta - \theta'}.
    \end{align*}
    for all $i \in 0, \dots, N$ and $\theta, \theta' \in \mathbb R^d$.
\end{ass}

Under this assumption the log-concavity constants, $m$ and $M$, of the posterior both increase linearly with $N$, as shown by the following Lemma.

\begin{lem}
    \label{lem:lin-just}
    Suppose Assumption \ref{ass:linear} holds. Then the log-posterior, $f$, satisfies the following
    \begin{align*}
        f(\theta) - f(\theta') - \nabla f(\theta')^\top (\theta - \theta') &\geq \frac{l(N+1)}{2} \norm{\theta - \theta'}^2 \\
        \norm{\nabla f(\theta) - \nabla f(\theta')} &\leq L(N+1) \norm{\theta - \theta'}.
    \end{align*}
    Thus the posterior is strongly log-concave with parameters $M=(N+1)L$ and $m=(N+1)l$.
\end{lem}

To see the potential benefit of using control variates to estimate the gradient in situations where $N$ is large, we can now compare the variance bound from Lemma \ref{lem:grad-noise}, with a bound on the variance of the simple estimator, $\nabla \hat f (\theta)$. If we assume that $\norm{\nabla f_i(\theta)}$ is bounded by some constant $\sigma$, for all $i = 0, \dots, N$ and for all $\theta \in \mathbb R^d$, then \cite{Dubey2016} show that for SGLD
\begin{equation}
    \E \norm{\nabla \hat f(\theta) - \nabla f(\theta)}^2 \leq \frac{2N^2 \sigma^2}{n},
\end{equation}
for all $\theta \in \mathbb R^d$.

We can see that the bound on the gradient estimate variance in \eqref{grad-var} depends on the distance between $\theta_k$ and $\hat \theta$. Appealing to the Bernstein-von Mises theorem \citep{LeCam2012}, under standard asymptotics we would expect the distance $\E \norm{\theta_k - \hat \theta}^2$ to be $O(1 / N)$, if $\hat \theta$ is within $O(N^{-1/2})$ of the posterior mean, once the MCMC algorithm has burnt in. As $M$ is $O(N)$, this suggests that using control variates could give an $O(N)$ reduction in variance, and this plays a key part in the computational cost improvements we show in the next section.

\subsection{Computational cost of SGLD-CV}
\label{sec:comp-cost-sgld}

In this section, we investigate how applying control variates to the gradient estimate of SGLD reduces the computational cost of the algorithm. 
    
In order to show this, we investigate the Wasserstein-Monge-Kantorovich (Wasserstein) distance $W_2$ between the distribution defined by the SGLD-CV algorithm at each iteration and the true posterior as $N$ is changed. For two measures $\mu$ and $\nu$ defined on the probability space $(\mathbb R^d, B(\mathbb R^d))$, and for a real number $q > 0$, the distance $W_q$ is defined by
    \begin{equation*}
        W_q(\mu, \nu) = \left[ \inf_{\gamma \in \Gamma(\mu, \nu)} \int_{\mathbb R^d \times \mathbb R^d} \norm{\theta - \theta'}^q d\gamma(\theta,\theta') \right]^{\frac 1 q},
    \end{equation*}
    where the infimum is with respect to all joint distributions $\Gamma$ having $\mu$ and $\nu$ as marginals. The Wasserstein distance is a natural distance measure to work with for Monte Carlo algorithms, as discussed in \cite{{Dalalyan2017, Durmus2016}}.
    
One issue when working with the Wasserstein distance is that it is not invariant to transformations. For example scaling all entries of $\theta$ by a constant will scale the Wasserstein distance by the same constant. A linear transformation of the parameters will result in a posterior that is still strongly log-concave, but with different constants $m$ and $M$. To account for this we suggest measuring error by the quantity $\sqrt m W_2$, which is invariant to scaling $\theta$ by a constant. Theorem 1 of \cite{Durmus2016} bounds the standard deviation of any component of $\theta$ by a constant times $1/\sqrt{m}$, so we can view the quantity  $\sqrt m W_2$ as measuring the error on a scale that is relative to the variability of $\theta$ under the posterior distribution.
    
There are a number of quantities that will affect the performance of SGLD and SGLD-CV. These include the step size $h$, the minibatch size $n$ and the total number of iterations $K$. In the analysis to follow we find conditions on $h$, $n$ and $K$ that ensure the Wasserstein distance between the distribution defined by SGLD-CV and the true posterior distribution $\pi$ are less than some $\epsilon > 0$. We use these conditions to calculate the computational cost, measured as $Kn$, required for this algorithm to reach the satisfactory error $\epsilon$.

The first step is to find an upper bound on the Wasserstein distance between SGLD-CV and the posterior distribution in terms of $h$, $n$, $K$ and the constants $m$ and $M$ declared in Assumption \ref{ass:stronvex}.

\begin{prop}
    \label{thm:wasser-bound}
    Under Assumptions \ref{ass:stronvex} and \ref{ass:Lipschitz}, let $\theta_{k}$ be the state of SGLD-CV at the $k^{th}$ iteration of the algorithm with stepsize $h$, initial value $\theta_0$, centering value $\hat \theta$. Let the distribution of $\theta_k$ be $\nu_k$. Denote the expectation of $\theta$ under the posterior distribution $\pi$ by $\bar \theta$. If $h < \frac{2m}{2M^2 + m^2}$, then for all integers $k \geq 0$,
\begin{equation*}
    W_2(\nu_k, \pi) \leq (1 - A)^K W_2(\nu_0, \pi) + \frac C A + \frac{B^2}{C + \sqrt A B},
\end{equation*}
where
\begin{align*}
    A &= 1 - \sqrt{\frac{2h^2M}{n} + (1 - mh)^2}, \\
    B &= \sqrt{\frac{2h^2M^2}{n} \left[ \E\norm{\hat \theta - \bar \theta}^2 + \frac d m \right]}, \\
    C &= \alpha M (h^3 d)^{\frac 1 2},
\end{align*}
$\alpha = 7 \sqrt 2 / 6$ and $d$ is the dimension of $\theta_{k}$.
\end{prop}

The proof of this proposition is closely related to the proof of Proposition 2 of \cite{Dalalyan2017}. The extra complication comes from our bound on the variance of the estimator of the gradient; which depends on the current state of the SGLD-CV algorithm, rather than being bounded by a global constant. 

We can now use Proposition \ref{thm:wasser-bound} to find conditions on $K$, $h$ and $n$ in terms of the constants $M$ and $m$ such that the Wasserstein distance is bounded by some positive constant $\epsilon_0/\sqrt{m}$ at its final iteration $K$. 

\begin{thm}
    \label{thm:conditions}
    Under Assumptions \ref{ass:stronvex} and \ref{ass:Lipschitz}, let $\theta_{K}$ be the state of SGLD-CV at the $K^{th}$ iteration of the algorithm with stepsize $h$, initial value $\theta_0$, centering value $\hat \theta$. Let the distribution of $\theta_k$ be $\nu_K$. Denote the expectation of $\theta$ under the posterior distribution $\pi$ by $\bar \theta$. Define $R := M / m$. Then for any $\epsilon_0 > 0$, if the following conditions hold:
\begin{align}
    h &\leq \frac 1 m \max \left\{ \frac{n}{2R^2 + n}, \frac{\epsilon_0^2}{64 R^2 \alpha^2 d} \right\}, \\
    Kh &\geq \frac 1 m \log\left[ \frac{4 m}{\epsilon_0^2} \left( \E\norm{\theta_0 - \bar \theta}_2^2 + d / m \right) \right], \\
    n &\geq \frac{64 R^2 \beta}{\epsilon_0^2} m \left[ \E \norm{\hat \theta - \bar \theta}^2 + \frac d m \right],
\end{align}
where 
\[
    \beta = \max \left\{ \frac{1}{2R^2 + 1}, \frac{\epsilon_0^2}{64 R^2 \alpha^2 d} \right\},
\]
$\alpha = 7 \sqrt 2 / 6$, and $d$ is the dimension of $\theta_k$,
then $W_2(\nu_K, \pi) \leq \epsilon_0/\sqrt{m}$.
\end{thm}

As a corollary of this result, we have the following, which gives a bound on the computational cost of SGLD, as measured by $Kn$, to achieve a required bound on the Wasserstein distance.
\begin{cor}
    \label{thm:scaling}
    Assume that Assumptions \ref{ass:stronvex} and \ref{ass:linear} and the conditions of Theorem \ref{thm:conditions} hold. Fix $\epsilon_0$ and define
    \[
        C_1 = \min \left\{ 2R^2 + 1, \frac{64 R^2 \alpha^2 d}{\epsilon_0^2} \right\}.
    \]
    and $C_2 := 64 R^2 \beta / \epsilon_0^2$.
    We can implement an SGLD-CV algorithm with $W_2(\nu_K, \pi) < \epsilon_0 / \sqrt m$ such that
    \[
        Kn \leq \left[ C_1 \log \left[ m \E \norm{\theta_0 - \bar \theta}^2 + d \right] + C_1 \log \frac{4}{\epsilon_0^2} + 1 \right] \left[ C_2 m \E \norm{\hat \theta - \bar \theta}^2 + C_2 d + 1 \right].
    \]
\end{cor}
The constants, $C_1$ and $C_2$, in the bound on $Kn$, depend on $\epsilon_0$ and $R=M/m$. It is simple to show that both constants are increasing in $R$. Under Assumption \ref{ass:linear} we have that $R$ is a constant as $N$ increases. Corollary \ref{thm:scaling} suggests that provided $\norm{\theta_0 - \bar \theta} < c / \sqrt m$ and $\norm{\hat \theta - \bar \theta} < c / \sqrt m$, for some constant $c$; then the computational cost of SGLD-CV will be bounded by a constant. Since we suggest starting SGLD-CV at $\hat \theta$, then technically we just need this to hold for $\norm{\hat \theta - \bar \theta}$. Under Assumption \ref{ass:linear} we have that $m$ increases linearly with $N$, so this corresponds to needing $\norm{\hat \theta - \bar \theta}<c_1/\sqrt{N}$ as $N$ increases. Additionally, by Theorem 1 of \cite{Durmus2016} we have that the variance of the posterior scales like $1/m=1/N$ as $N$ increases, so we can interpret the $1/\sqrt{N}$ factor as being a measure of the spread of the posterior as $N$ increases. The form of the corollary makes it clear that a similar argument would apply under weaker assumptions than Assumption \ref{ass:linear}. We only need that the ratio of the log-concavity constants, $M/m$, of the posterior remains bounded as $N$ increases.

This corollary also gives insight into the computational penalty you pay for a poor choice of $\theta_0$ or $\hat \theta$. As $\norm{\theta_0-\bar \theta}$ increases, the bound on the computational cost will increase logarithmically with this distance. By comparison the bound increases linearly with $\norm{\hat\theta-\bar \theta}$.

\subsection{Setup Costs}
\label{sec:setup}

There are a number of known results on the convergence of SGD under the strongly log-concave conditions of Assumption \ref{ass:stronvex}. These will allow us to quantify the setup cost of finding the point $\hat \theta$ in this setting. More complex cases are explored empirically in the experiments in Section \ref{sec:experiments}. Lemma \ref{lem:sgd} due to \cite{Nemirovski2009} quantifies the convergence of the final point of SGD.

\begin{lem}
    \label{lem:sgd}
    Under Assumption \ref{ass:stronvex}, let $\hat \theta$ denote the final state of SGD with stepsizes $h_k = 1 / (mk)$ after $K$ iterations. Suppose $\E \norm{\nabla \hat f(\theta)}^2 \leq D^2$ and denote the true mode of $f$ by $\theta^*$. Then it holds that
    \[
        \E \norm{\hat \theta - \theta^*}^2 \leq \frac{4D^2}{m^2 K}.
    \]
\end{lem}

If we again assume that, as in \cite{Dubey2016}, $\norm{\nabla f_i(\theta)}$ is bounded by some constant $\sigma$, for all $i = 0, \dots, N$ and $\theta \in \mathbb R^d$ then $D^2$ will be $O(N^2 / n)$.
This means that under Assumption \ref{ass:linear}, we will need to process the full dataset once before the SGD algorithm has converged to a mode $\hat \theta$ within $O(N^{-\frac 1 2})$ of the posterior mean. It follows that, for these cases there are two one off $O(N)$ setup costs, one to find an acceptable mode $\hat \theta$ and one to find the full log posterior gradient at this point $\nabla f(\hat \theta)$.

\section{Post-processing control variates}
\label{sec:post}

Control variates can also been used to improve the inferences made from MCMC by reducing the variance of the output directly. The general aim of MCMC is to estimate expectations of functions, $g(\theta)$, with respect to the posterior $\pi$. Given an MCMC sample $\theta^{(1)}, \dots, \theta^{(M)}$, from the posterior $\pi$, we can estimate $\E[g(\theta)]$ unbiasedly as
\[
    \E[g(\theta)] \approx \frac{1}{M} \sum_{i=1}^M g(\theta^{(i)}).
\]

Suppose there exists a function $h(\theta)$, which has expectation 0 under the posterior. We can then introduce an alternative function, \[\tilde g( \theta ) = g(\theta) + h(\theta),\] where $\E[\tilde g(\theta)] = \E[g(\theta)]$. If $h(\cdot)$ is chosen so that it is negatively correlated with $g(\theta)$, then the variance of $\tilde g(\theta)$ will be reduced considerably.

\citet{Mira2013} introduce a way of choosing $h(\theta)$ almost automatically by using the gradient of the log-posterior. Choosing $h(\cdot)$ in this manner is referred to as a zero variance (ZV) control variate. \citet{Friel2016} showed that, under mild conditions, we can replace the log-posterior gradient with an unbiased estimate and still have a valid control variate. SGMCMC methods produce unbiased estimates of the log-posterior gradient, and so it follows that these gradient estimates can be applied as ZV control variates. For the rest of this section, we focus our attention on SGLD, but these ideas are easily extendable to other stochastic gradient MCMC algorithms. We refer to SGLD with these post-processing control variates as SGLD-ZV.

Given the setup outlined above, \citet{Mira2013} propose the following form for $h( \theta )$, 
\[
    h(\theta) = \Delta Q(\theta) + \nabla Q( \theta ) \cdot \mb z,
\]
here $Q(\theta)$ is a polynomial of $\theta$ to be chosen and $\mb z = f( \theta ) / 2$. $\Delta$ refers to the \emph{Laplace operator} $\frac{ \partial^2 }{ \partial \theta_1^2 } + \dots + \frac{ \partial^2 }{ \partial \theta_d^2 }$. In order to get the best variance reduction, we simply have to optimize the coefficients of the polynomial $Q(.)$. In practice, first or second degree polynomials $Q(\theta)$ often provide good variance reduction \citep{Mira2013}. For the rest of this section we focus on first degree polynomials, so $Q(\theta) = \mb a^T \theta$, but the ideas are easily extendable to higher orders \citep{Papamarkou2014}.

The SGLD algorithm only calculates an unbiased estimate of $\nabla f(\theta)$, so we propose replacing $h(\theta)$ with the unbiased estimate
\begin{equation}
    \hat h (\theta) = \Delta Q(\theta) + \nabla Q(\theta) \cdot \mb{ \hat z }, 
\end{equation}
where $\mb{\hat z} = \nabla \hat f(\theta) / 2$. By identical reasoning to \citet{Friel2016}, $\hat h( \theta)$ is a valid control variate. Note that $\hat{ \mb z }$ can use any unbiased estimate, and as we will show later, the better the gradient estimate, the better this control variate performs. 

We set $Q(\theta)$ to be a linear polynomial $\mb a^T \theta$, so our SGLD-ZV estimate will take the following form
\begin{equation}
    \hat g(\theta) = g(\theta) + \mb a^T \mb{ \hat z}.
    \label{eq:zv-cv}
\end{equation}
Similar to standard control variates \citep{Ripley2009}, we need to find optimal coefficients $\mb{ \hat a }$ in order to minimize the variance of $\tilde g(\cdot)$, defined in \eqref{eq:zv-cv}. In our case, the optimal coefficients take the following form \citep{Friel2016}
\[
    \hat{ \mb a} = \V^{-1} \left( \mb{ \hat z } \right) \C\left( \mb{ \hat z }, g(\theta) \right).
\]
This means that SGLD already calculates all the necessary terms for these control variates to be applied for free. So the post-processing step can simply be applied once when the SGLD algorithm has finished, provided the full output plus gradient estimates are stored. With this in place, we can write down the full algorithm in the linear case, which is given in Algorithm \ref{alg:SGLD-ZV}. For higher order polynomials, the calculations are much the same, but more coefficients need to be estimated \citep{Papamarkou2014}.

\begin{algorithm}[t]
    \caption{SGLD-ZV}
    \begin{algorithmic}
        \Require $\{ \theta_k, \nabla \hat f(\theta_k)  \}_{k=1}^K$ \Comment{SGLD output}
        \State Set $\mb z_k \leftarrow  \frac 1 2 \nabla \hat f(\theta_k) $
        \State Estimate $V_{\mb z} \leftarrow \V(\mb z)$, $C_{g,\mb z} \leftarrow \text{Cov}(g(\theta), \mb z)$
        \State $\hat{ \mb a}_{j} \leftarrow \left[ V_{\mb z} \right]^{-1} C_{g,\mb z}$
        \For{$k = 1 \dots K$}
        \State $\hat g(\theta_k) \leftarrow g(\theta_k) + \hat{\mb a}^T \mb z_k$
        \EndFor
    \end{algorithmic}
    \label{alg:SGLD-ZV}
\end{algorithm}

The efficiency of ZV control variates in reducing the variance of our MCMC sample is directly affected by using an estimate of the gradient rather than the truth. For the remainder of this section, we investigate how the choice of the gradient estimate, and the minibatch size $n$, affects the variance reduction.

\begin{ass}
    \label{ass:bounded}
    $\V[ \phi(\theta)] < \infty$ and $\V[ \hat \psi (\theta)] < \infty$. $\E_{\theta | \mb x} \norm{\nabla f_i(\theta)}^2$ is bounded by some constant $\sigma$ for all $i = 0, \dots N$, $\theta \in \mathbb R^d$.
\end{ass}

\begin{thm}
    Under Assumption \ref{ass:bounded},
    define the optimal variance reduction for ZV control variates using the full gradient estimate to be $R$, and the optimal variance reduction using SGLD gradient estimates to be $\hat R$. Then we have that
    \begin{equation}
        \hat R \geq \frac{ R }{ 1 + [\sigma (N + 1)]^{- 1} \E_{\theta | \mb x}[ \E_S \norm{ \xi_S(\theta) }^2 ]},
        \label{eq:var-reduc}
    \end{equation}
    where $\xi_S(\theta)$ is the noise in the log-posterior gradient estimate.
\label{thm:var-reduc}
\end{thm}
The proof of this result is given in the Appendix. An important consequence of Theorem \ref{thm:var-reduc} is that if we use the standard SGLD gradient estimate, then the denominator of \eqref{eq:var-reduc} is $O(n/N)$, so our variance reduction diminishes as $N$ gets large. However, if we use the SGLD-CV estimate instead (the same probably holds for other control variate algorithms such as SAGA), then under standard asymptotics, the denominator of \eqref{eq:var-reduc} is $O(n)$, so the variance reduction does not diminish with increasing dataset size. It follows that for best results, we recommend using the ZV post-processing step after running the SGLD-CV algorithm, especially for large $N$. The ZV post-processing step can be immediately applied in exactly the same way to other stochastic gradient MCMC algorithms, such as SGHMC and SGNHT \citep{{Chen2014,Ding2014}}.

It is worth noting that there are some storage constraints for SGLD-ZV. This algorithm requires storing the full MCMC chain, as well as the gradient estimates at each iteration. So the storage cost is twice the storage cost of a standard SGMCMC run. However, in some high dimensional cases, the required SGMCMC test statistic is estimated on the fly using the most recent output of the chain and thus reducing the storage costs. We suggest that if the dimensionality is not too high, then the additional storage cost of recording the gradients to apply the ZV post-processing step can offer significant variance reduction for free. However, for very high dimensional parameters, the cost associated with storing the gradients may preclude the use of the ZV step.

\section{Experiments}
\label{sec:experiments}

\subsection{Logistic regression}
\label{sec:lreg}

\begin{figure}[t]
    \centering
    \includegraphics[width=300px]{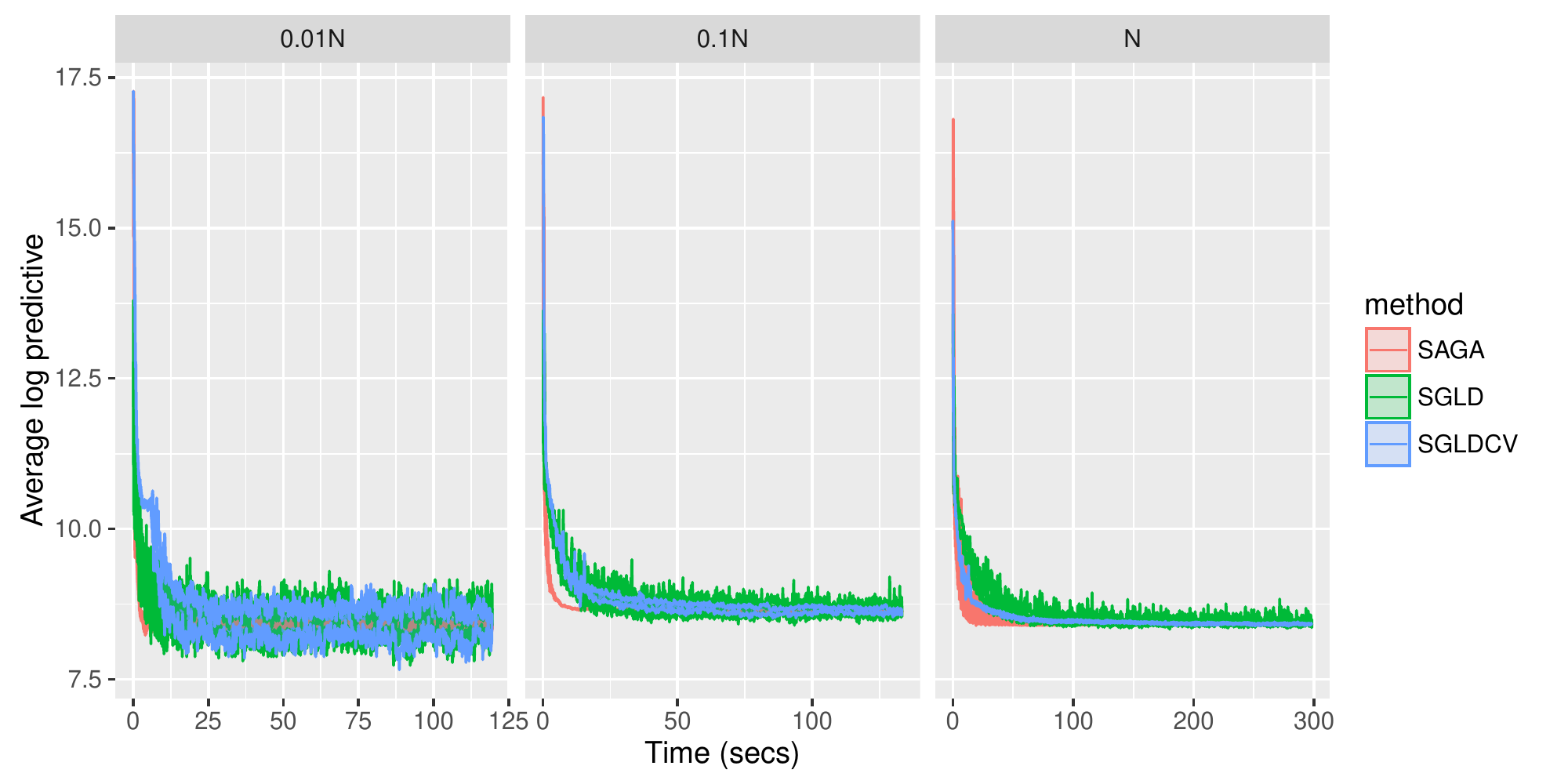}
    \caption{Log predictive density over a test set every 10 iterations of SGLD, SGLD-CV and SAGA fit to a logistic regression model as the data size $N$ is varied.}
    \label{fig:logistic-cv}
\end{figure}

\begin{figure}[t]
    \centering
    \includegraphics[width=300px]{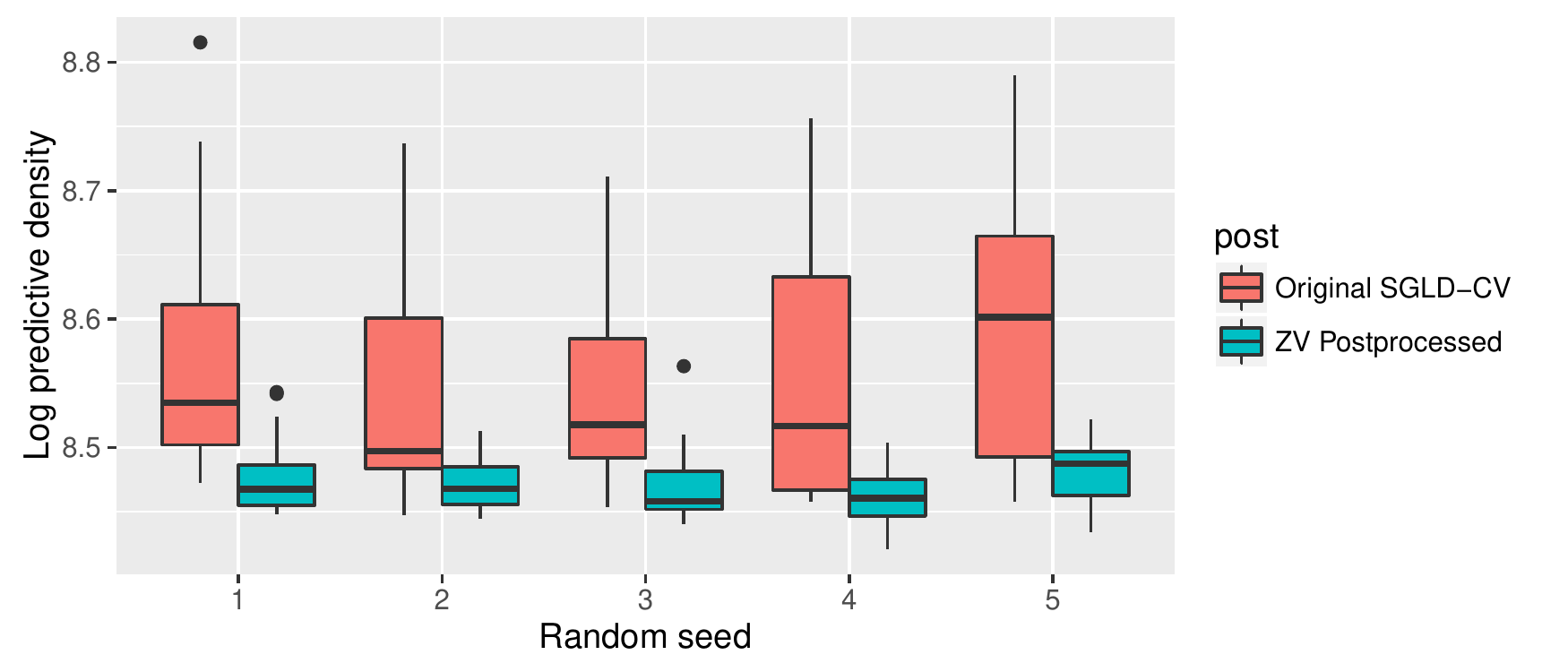}
    \caption{Plots of the log predictive density of an SGLD-CV chain when ZV post-processing is applied versus when it is not, over 5 random runs. Logistic regression model on the cover type dataset \citep{Blackard1999}.}
    \label{fig:logistic-zv}
\end{figure}

We examine our approaches on a Bayesian logistic regression problem. The probability of the $i^{th}$ output $y_i \in \{ -1, +1 \}$ is given by
\[
    p( y_i | x_i, \beta ) = \frac{1}{1 + \exp(-y_i \beta^{T} x_i)}.
\]
We use a Laplace prior for $\beta$ with scale 1.

We used the cover type dataset \citep{Blackard1999}, which has 581012 observations, which we split into a training and test set. First we run SGLD, SGLD-CV and SAGA on the dataset, all with minibatch size 500. The method SAGA was discussed at the end of Section \ref{sec:contr-vari-grad}. To empirically support the scalability results of Theorem \ref{thm:scaling}, we fit the model 3 times. In each fit, the dataset size is varied, from about 1\% of the full dataset to the full dataset size $N$. The performance is measured by calculating the log predictive density on a held-out test set every 10 iterations. Some of our examples are high dimensional, so our performance measure aims to reduce dimensionality while still capturing important quantities such as the variance of the chain. We include the burn-in of SGLD and SAGA, to contrast with the optimisation step required for SGLD-CV which is included in the total computational time.

The results are plotted against time in Figure \ref{fig:logistic-cv}. The results illustrate the efficiency gains of SGLD-CV over SGLD as the dataset size increases, as expected from Theorem \ref{thm:scaling}. SAGA outperforms SGLD-CV in this example because SGLD converges quickly in this simple setting. In the more complicated examples to follow, we show that SAGA can be slow to converge.

We also compare the log predictive density over a test set for SGLD-CV with and without ZV post-processing, averaged over 5 runs at different seeds. We apply the method to SGLD-CV rather than SGLD due to the favourable scaling results as discussed after Theorem \ref{thm:var-reduc}. Results are given in Figure \ref{fig:logistic-zv}. The plot shows box-plots of the log predictive density of the SGLD sample before and after post-processing using ZV control variates. The plots show excellent variance reduction of the chain.

\subsection{Probabilistic matrix factorization}
\label{sec:matfac}

\begin{figure}[t]
    \centering
    \includegraphics[width=300px]{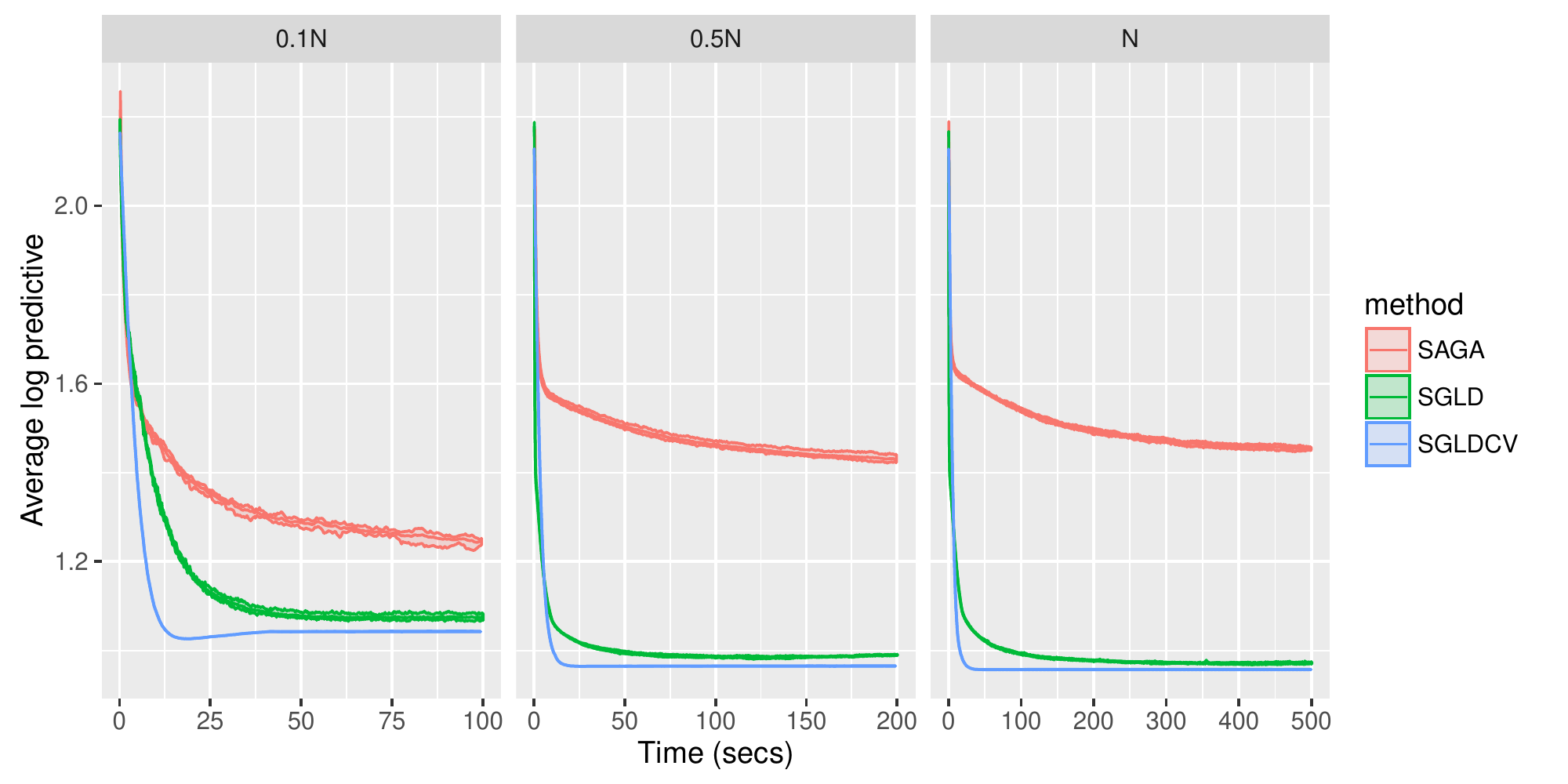}
    \caption{Log predictive density over a test set of SGLD, SGLD-CV and SAGA fit to a Bayesian probabilistic matrix factorization model as the number of users is varied, averaged over 5 runs. We used the Movielens ml-100k dataset.}
    \label{fig:matFac-cv}
\end{figure}

\begin{figure}[t]
    \centering
    \includegraphics[width=300px]{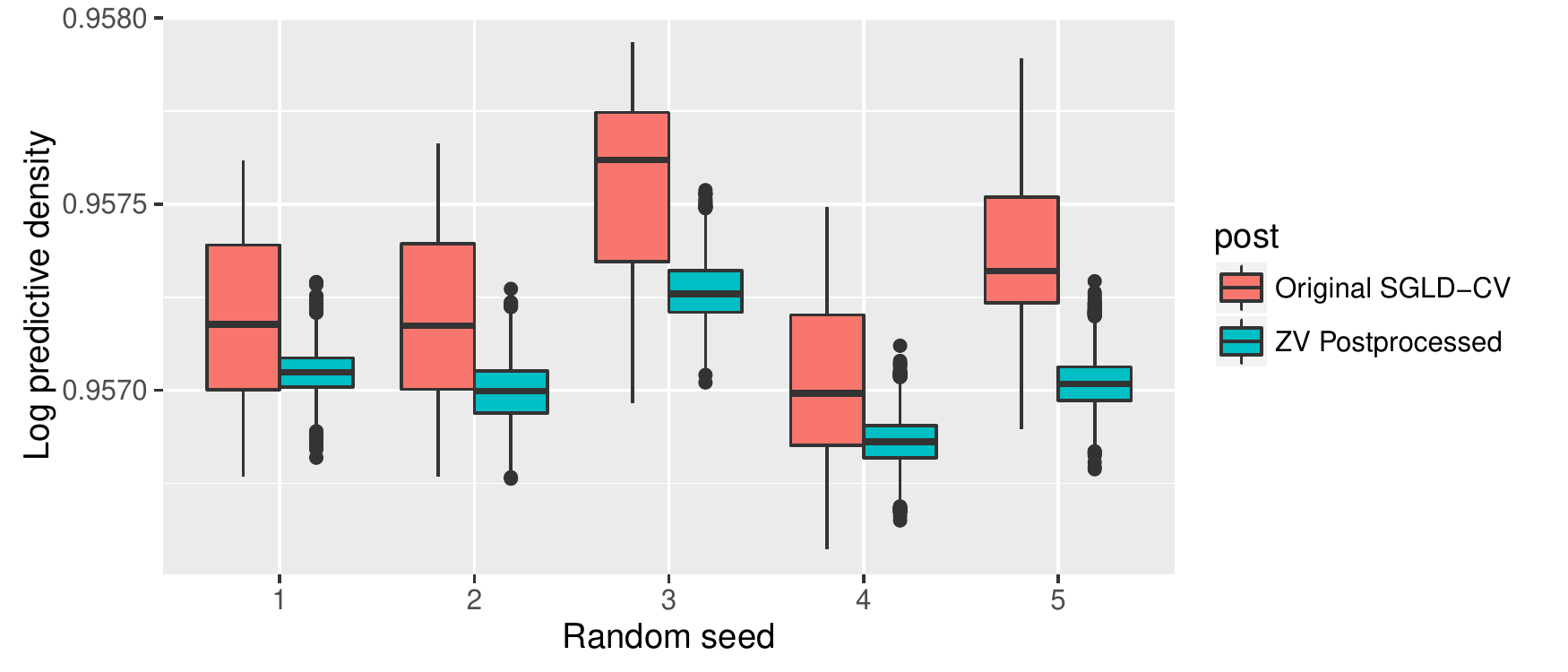}
    \caption{Plots of the log predictive density of an SGLD-CV chain when ZV post-processing is applied versus when it is not, over 5 random runs. SGLD-CV algorithm applied to a Bayesian probablistic matrix factorization problem using the Movielens ml-100k dataset.}
    \label{fig:matfac-zv}
\end{figure}

A common recommendation system task is to predict a user's rating of a set of items, given previous ratings and the ratings of other users. The end goal is to recommend new items that the user will rate highly. Probabilistic matrix factorization (PMF) is a popular method to train these models \citep{Mnih2008}. As the matrix of ratings is sparse, over-fitting is a common issue in these systems, and Bayesian approaches are a way to account for this \citep{Ahn2015}.

In this experiment, we apply SGLD, SGLD-CV and SAGA to a Bayesian PMF problem, using a model similar to \citet{Ahn2015} and \citet{Chen2014}. We use the Movielens dataset ml-100k\footnote{\label{ml-100k}https://grouplens.org/datasets/movielens/100k/}, which contains 100,000 ratings from almost 1,000 users and 1,700 movies. We use batch sizes of 5,000, with a larger minibatch size chosen due to the high-dimensional parameter space. As before, we compare performance by calculating the predictive distribution on a held out dataset every 10 iterations.

We investigate the scaling results of SGLD-CV and SAGA versus SGLD by varying the dataset size. We do this by limiting the number of users in the dataset, ranging from 100 users to the full 943. The results are given in Figure \ref{fig:matFac-cv}. Once again the scaling improvements of SGLD-CV as the dataset size increases are clear. 

In this example SAGA converges slowly in comparison even to SGLD. In fact the algorithm converges slowly in all our more complex experiments. The problem is particularly bad for large $N$. This is likely a result of the starting point for SAGA being far from the posterior mode. Empirically, we found that the gradient direction and magnitude can update very slowly in these cases. This is not an issue for simpler examples such as logistic regression, but for more complex examples we believe it could be a sign that the algorithm is getting stuck in, or moving slowly through, local modes where the gradient is comparatively flatter. The problem appears to be made worse for large $N$ when it takes longer to update $g_\alpha$. This is an example where the optimisation step of SGLD-CV is an advantage, as the algorithm is immediately started close to the posterior mode and so the efficiency gains are quickly noted. This issue with SAGA could be related to the starting point condition for SGLD-CV as detailed in Corollary \ref{thm:scaling}. Due to the form of the Wasserstein bound, it is likely that SAGA would have a similar starting point condition.

Once again we compare the log predictive density over a test set for SGLD-CV with and without ZV post-processing when applied to the Bayesian PMF problem, averaged over 5 runs at different seeds. Results are given in Figure \ref{fig:logistic-zv}. The plot shows box-plots of the log predictive density of the SGLD sample before and after post-processing using ZV control variates. The plots show excellent variance reduction of the chain.

\subsection{Latent Dirichlet allocation}
\label{sec:lda}

\begin{figure}[t]
    \centering
    \includegraphics[width=300px]{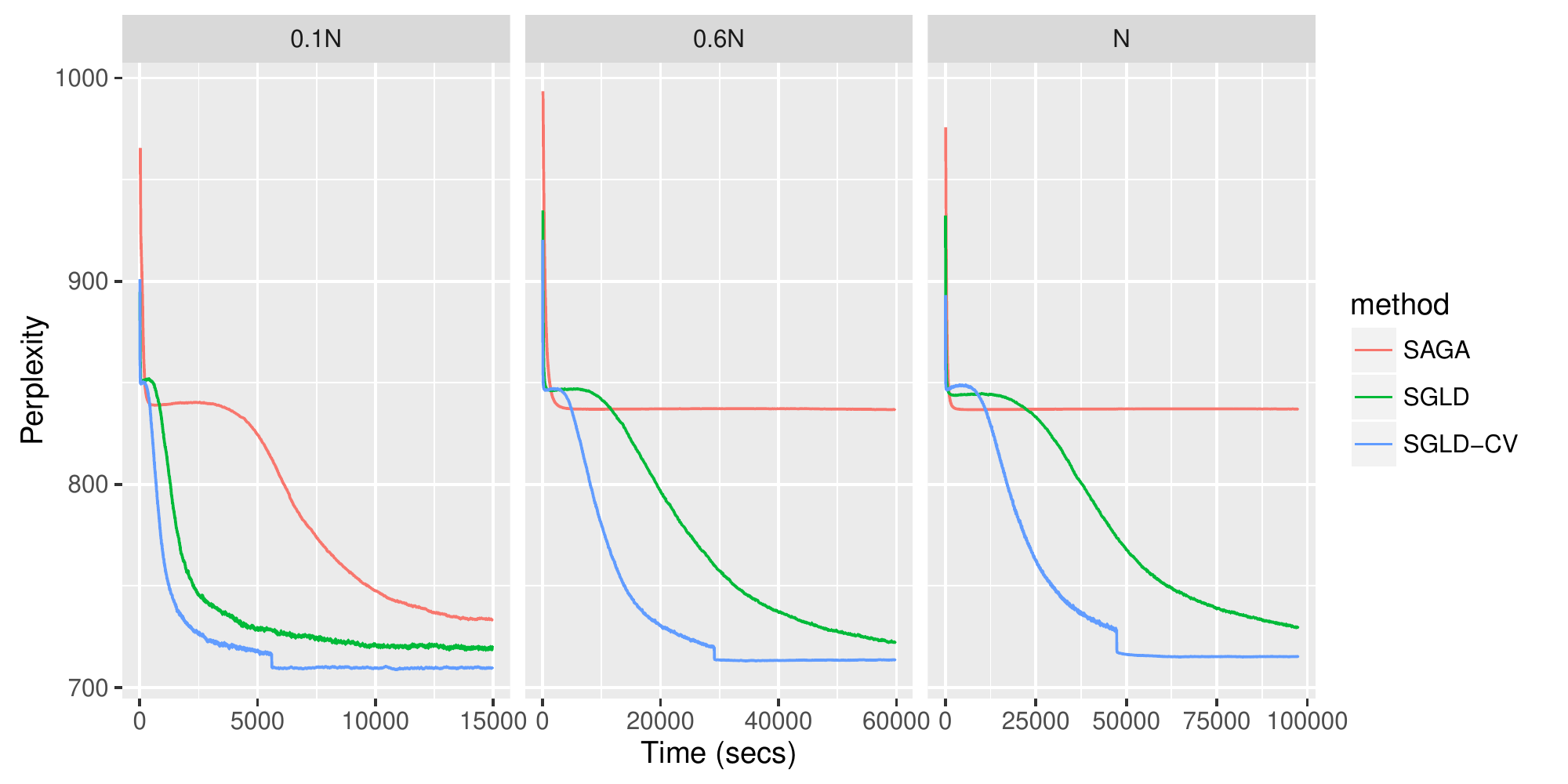}
    \caption{Perplexity of SGLD and SGLD-CV fit to an LDA model as the data size $N$ is varied, averaged over 5 runs. The dataset consists of scraped Wikipedia articles.}
    \label{fig:lda}
\end{figure}

Latent Dirichlet allocation (LDA) is an example of a topic model used to describe collections of documents by sets of discovered topics \citep{Blei2003}. The input consists of a matrix of word frequencies in each document, which is very sparse, motivating the use of a Bayesian approach to avoid over-fitting.

Due to storage constraints, it was not feasible to apply SGLD-ZV to this problem, so we focus on SGLD-CV. We scraped approximately 80,000 documents from Wikipedia, and used the 1,000 most common words to form our document-word matrix input. We used a similar formulation to \citet{Patterson2013}, though we did not use a Riemannian sampler. 

Once again in our comparison of SGLD, SGLD-CV and SAGA, we vary the dataset size, this time by changing the number of documents used in fitting the model, from 10,000 to the full 81,538. We use batch sizes of 50 documents. We measure the performance of LDA using the \emph{perplexity} on held out words from each document, a standard performance measure for this model. The results are given in Figure \ref{fig:lda}. Here the scalability improvements of using SGLD-CV over SGLD are clear as the dataset size increases. This time the batch size is small compared to the dataset size, which probably makes the scalability improvements more obvious. The sudden drop in perplexity for the SGLD-CV plot occurs at the switch from the stochastic optimization step to SGLD-CV. This is likely a result of the algorithm making efficient use of the Gibbs step to simulate the latent topics.

An interesting aspect of this problem is that it appears to have a pronounced local mode where each of the methods become trapped (this can be seen by the blip in the plot at a perplexity of around 850). SGLD-CV is the first to escape followed by SGLD, but SAGA takes a long time to escape. This is probably due to a similar aspect as the one discussed in the previous experiment (Section \ref{sec:matfac}). Similar to the previous experiment, we find that while SAGA seems trapped, its gradient estimate changes very little, which could be a sign that the algorithm is moving very slowly through an area with a relatively flat gradient, such as a local mode. A simple solution would be to start SAGA closer to the mode using a stochastic optimisation scheme.

\section{Discussion}

    We have used control variates for stochastic gradient MCMC to reduce the variance in the gradient estimate. We have shown that in the strongly log-concave setting, and under standard asymptotics, this proposed SGLD-CV algorithm reduces the computational cost of stochastic gradient Langevin dynamics to $O(1)$. Our theoretical results give results on the computational cost under non-standard asymptotics also, and show there should be some benefit provided distance between the centering value $\hat \theta$ and the posterior mean inversely depends on $N$. The algorithm relies on a setup cost that estimates the posterior mode which replaces the burn-in of SGLD. We have explored the cost of this step both theoretically and empirically. We have empirically supported these scalability results on a variety of interesting and challenging problems from the statistics and machine learning literature using real world datasets. The simulation study also revealed that SGLD-CV was less susceptible to getting stuck in local stationary points than an alternative method that performs variance reduction using control variates, SAGA \citep{Dubey2016}. An interesting future extension would be to reduce the startup cost of SGLD-CV, along with introducing automatic step-size tuning.

We showed that stochastic gradient MCMC methods calculate all the information needed to apply zero variance post-processing control variates. This improves the inference of the output by reducing its variance. We explored how the variance reduction is affected by the minibatch size and the gradient estimate, and show using SGLD-CV or SAGA rather than SGLD can achieve a better variance reduction. We demonstrated this variance reduction empirically. A limitation of these post-processing control variates is they require the whole chain, which can lead to high storage costs if the dimensionality of the sample space is high. Future work could explore ways to reduce the storage costs of stochastic gradient MCMC.

\section{Acknowledgements}
\label{sec:acknowledgements}

The first author gratefully acknowledges the support of the EPSRC funded EP/L015692/1 STOR-i Centre for Doctoral Training. This work was supported by EPSRC grant EP/K014463/1, ONR Grant N00014-15-1-2380 and NSF CAREER Award IIS-1350133.

\bibliography{submission}
\bibliographystyle{apa}

\newpage
\appendix
\onecolumn

\section{Computational cost proofs}

\subsection{Proof of Proposition \ref{thm:wasser-bound}}

\begin{proof}
Let $\pi$ be the invariant distribution of the underlying dynamics, so that it has density $e^{-f(\theta)} = p(\theta | \mb x)$, and define $W_2(\nu_k, \pi)$ to be the Wasserstein distance between $\nu_k$ and $\pi$. Define $\xi_{k}$ to be the SGLD-CV gradient noise term. Then we can write a single step of SGLD-CV as
\[
    \theta_{k+1} = \theta_{k} + h \nabla f(\theta_{k}) + h \xi_k + \sqrt{2h} \zeta_k,
\]

We have that $\theta_k \sim \nu_k$, and follow similarly to the proof of \citet[][Proposition 2]{Dalalyan2017}. First define $Y_0$ to be a draw from the invariant distribution $\pi$, such that the joint distribution of $Y_0$ and $\theta_k$ minimizes $\E\norm{Y_0 - \theta_k}^2$. Here $\norm{.}$ denotes the Euclidean distance for $\mathbb R^d$. It follows that $\E\norm{Y_0 - \theta_k}^2 = W_2^2(\nu_k, \pi)$.

Let $B_t$ be a $d$-dimensional Wiener process, independent of $\theta_k$, $Y_0$ and $\xi_k$ but which we couple to the injected noise $\zeta_k$ so that $B_h = \sqrt h \zeta_k$. Now let $Y_t$, $t > 0$, follow the diffusion
\begin{equation}
    Y_t = Y_0 + \int_0^t \nabla f(Y_s) ds + \sqrt 2 B_t.
\end{equation}
Let $\Delta_k = Y_0 - \theta_k$ and $\Delta_{k+1} = Y_h - \theta_{k+1}$. Since we started the process $Y_t$ from $Y_0 \sim \pi$, then it follows that $Y_t \sim \pi$ for all $t > 0$. Also since $W_2^2(\nu_{k+1}, \pi)$ minimizes the expected squared distance between two random variables with marginals $\nu_{k+1}$ and $\pi$ then it follows that $W_2^2(\nu_{k+1}, \pi) \leq \E\norm{\Delta_{k+1}}^2$.

Let us define
\begin{align}
    U &= \nabla f(\theta_k + \Delta_k) - \nabla f(\theta_k),
    \label{eq:U} \\
    V &= \int_0^h \left[ \nabla f(Y_t) - \nabla f(Y_0) \right] dt.
    \label{eq:V}
\end{align}
Then by the unbiasedness of the gradient estimation, $\xi_k$ has mean 0 regardless of the value of $\theta_k$. Thus
\begin{align*}
    \E \norm{\Delta_{k+1}}^2 &= \E \norm{\Delta_k + hU + V}^2 + \E\norm{\xi_k}^2 \\
    &\leq \left[\E\norm{\Delta_k - h U} + \E \norm{V} \right]^2 + h^2 \E\norm{\xi_k}^2.
\end{align*}
We can then apply Lemmas 2 and 4 in \cite{Dalalyan2017}, stated below in Lemmas \ref{lem:dal-1} and \ref{lem:dal-3}, as well as applying the gradient noise bound in Lemma \ref{lem:grad-noise}, to obtain a bound on $W_2^2(\nu_{k+1}, \pi)$ given $W_2^2(\nu_k, \pi)$.
\begin{lem}
    \label{lem:dal-1}
    With $U$ as defined in \eqref{eq:U}, if $h < 2m / (2M^2 + m^2)$, then $\norm{\Delta_k - hU} \leq (1 - mh) \norm{\Delta_k}$.
\end{lem}
The original lemma by \cite{Dalalyan2017} assumed $h < 2 / (m + M)$, but this holds when $h < 2m / (2M^2 + m^2)$ as $m \leq M$.

\begin{lem}
    \label{lem:dal-3}
    Under Assumption \ref{ass:stronvex}. Let $V$ be as defined in \eqref{eq:V}, then
    \[
        \E \norm{V} \leq \frac 1 2 (h^4 M^3 d)^{\frac 1 2} + \frac 2 3 (2 h^3 d)^{\frac 1 2} M.
    \]
\end{lem}

Finally we can apply Lemma \ref{lem:grad-noise}, as stated in the main body, to get
\begin{align*}
    \E\norm{\xi_k}^2 
    &\leq \frac{M^2}{n} \E \norm{\theta_k - \hat \theta}^2 \\
    &\leq \frac{2M^2}{n} \E \norm{\theta_k - Y_0}^2 + \frac{2M^2}{n^2} \E \norm{Y_0 - \hat \theta}^2 \\
    &\leq \frac{2M^2}{n} W_2^2(\nu_k, \pi) + \frac{2M^2}{n} \E \norm{Y_0 - \hat \theta}^2
\end{align*}
Using Theorem 1 of \cite{Durmus2016}
\begin{equation}
    \E \norm{Y_0 - \hat \theta}^2 \leq \E\norm{\hat \theta - \bar \theta}^2 + \frac d m.
\end{equation}
It follows that
\begin{equation}
    \E\norm{\xi_k}^2 \leq \frac{2M^2}{n} W_2^2(\nu_k, \pi) + \frac{2M^2}{n} \left[ \E\norm{\hat \theta - \bar \theta}^2 + \frac d m \right].
\end{equation}

Now using that $W_2^2(\nu_{k+1}, \pi) \leq \E\norm{\Delta_{k+1}}^2$ we get the following
\[
    W_2^2(\nu_{k+1}, \pi) \leq \left[ (1 - mh) W_2(\nu_k, \pi) + \alpha M (h^3 d)^{\frac 1 2} \right]^2 + \frac{2h^2M^2}{n} W_2^2(\nu_k, \pi) + \frac{2h^2M^2}{n} \left[ \E\norm{\hat \theta - \bar \theta}^2 + \frac d m \right],
\]
where $\alpha = 7 \sqrt 2 / 6$. Gathering like terms we can further bound $W_2^2(\nu_{k+1}, \pi)$ to get the following recursive formula
\begin{align*}
    W_2^2(\nu_{k+1}, \pi) &\leq \left[ (1 - A) W_2(\nu_k, \pi) + C \right]^2 + B^2
\end{align*}
where
\begin{align*}
    A &= 1 - \sqrt{\frac{2h^2M^2}{n} + (1 - mh)^2} \\
    B &= \sqrt{\frac{2h^2M^2}{n} \left[ \E\norm{\hat \theta - \bar \theta}^2 + \frac d m \right]} \\
    C &= \alpha M (h^3 d)^{\frac 1 2}.
\end{align*}
We can now apply Lemma 1 of \cite{Dalalyan2017}, as stated below to solve this recurrence relation.

\begin{lem}
    \label{lem:recur}
    Let $A$, $B$ and $C$ be non-negative numbers such that $A \in (0,1)$. Assume that the sequence of non-negative numbers $x_k$, $k = 0, 1, \dots$, satisfies the recursive inequality
    \[
        x_k^2 \leq \left[(1-A)x_k + C\right]^2 + B^2
    \]
    for every integer $k > 0$. Then for all integers $k \geq 0$
    \[
        x_k \leq (1-A)^k x_0 + \frac C A + \frac{B^2}{C + \sqrt A B}
    \]
\end{lem}

To complete the proof all that remains is to check $A \in (0,1)$ so that Lemma \ref{lem:recur} can be applied. Clearly $A < 1$, since $n \geq 1$ we have
\[
    A \geq 1 - \sqrt{2h^2M^2 - (1 - mh)^2},
\]
and the RHS is positive when $h \in (0, 2m /(2M^2 + m^2))$.
\qed
\end{proof}

\subsection{Proof of Theorem \ref{thm:conditions}}

\begin{proof}
    Starting from Proposition \ref{thm:wasser-bound}, we have that
\begin{equation}
    W_2(\nu_K, \pi) \leq (1 - A)^K W_2(\nu_0, \pi) + \frac C A + \frac{B^2}{C + \sqrt A B}.
    \label{eq:wasser-bound}
\end{equation}
where
\begin{align*}
    A = 1 - \sqrt{\frac{2h^2M^2}{n} + (1 - mh)^2},&\qquad
    B = \sqrt{\frac{2h^2M^2}{n} \left[ \E\norm{\hat \theta - \bar \theta}^2 + \frac d m \right]}, \qquad
    C = \alpha M (h^3 d)^{\frac 1 2},
\end{align*}

Suppose we stop the algorithm at iteration $K$. Using \eqref{eq:wasser-bound}, the following are sufficient conditions that ensure $W_2^2(\nu_K, \pi) < \epsilon_0 / \sqrt m$,
\begin{align}
    (1 - A)^K W_2(\nu_0, \pi) \leq \frac{\epsilon_0}{2 \sqrt m},
    \label{eq:kh-bound} \\
    \frac C A \leq \frac{\epsilon_0}{4 \sqrt m},
    \label{eq:c-a} \\
    \frac{B^2}{C + \sqrt A B} \leq \frac{\epsilon_0}{4 \sqrt m}.
    \label{eq:b-sqrta}
\end{align}

The starting point $\theta_0$ is deterministic, so from Theorem 1 of \cite{Durmus2016}
\begin{equation}
    W_2^2(\nu_0, \pi) \leq \E\norm{\theta_0 - \bar \theta}^2 + \frac d m.
    \label{eq:degenerate}
\end{equation}

If we rewrite
\begin{equation}
    h = \frac{\gamma}{m} \left[ \frac{2n}{2R^2 + n} \right],
\end{equation}
where $\gamma \in (0,1)$ is some constant and $R := M / m$ as defined in the theorem statement, then it follows that we can write
\begin{equation}
    A = 1 - \sqrt{1 - 2mh(1 - \gamma)}.
    \label{eq:A0}
\end{equation}
Since we have the condition
\[
    h \leq \frac 1 m \left[ \frac{n}{2R^2 + n} \right],
\]
then $\gamma \leq \frac 1 2$.

Now suppose, using \eqref{eq:A0}, we set
\begin{equation}
    Kh \geq \frac 1 m \log\left[ \frac{4 m}{\epsilon_0^2} \left( \E\norm{\theta_0 - \bar \theta}_2^2 + d / m \right) \right]
    \label{eq:kh-cond}
\end{equation}
Then using the result for the deterministic starting point $\theta_0$ \eqref{eq:degenerate}, we find that \eqref{eq:kh-cond} implies that
\begin{align*}
    \frac{\epsilon_0}{2 \sqrt m} &\geq \exp\left[ - mhK / 2\right] \sqrt{ \E \norm{\theta_0 - \bar \theta}^2 + \frac d m } \\
    &\geq \left[ 1 - mh \right]^{\frac K 2} W_2(\nu_0, \pi) \\
    &\geq (1 - A)^K W_2(\nu_0, \pi),
\end{align*}
Using \eqref{eq:A0} and that our conditions imply $\gamma < 1/2$. Hence \eqref{eq:kh-bound} holds.

Using that for some real number $y \in [0,1]$, $\sqrt{1 - y} \leq 1 - y/2$, we can bound $A$ by
\begin{equation}
    A \geq 1 - \sqrt{1 - 2mh(1 - \gamma)} \geq mh(1 - \gamma) := A_0.
\end{equation}

As $\gamma \leq 1/2$, for \eqref{eq:c-a} to hold it is sufficient that
\[
\frac{\epsilon_0}{4\sqrt m} \geq \frac{C}{A_0},
\]
where $C/A_0 \geq 2\alpha M \sqrt{hd}/m$. This leads to the following sufficient condition on $h$, 
%
\begin{equation}
    h \leq \frac{1}{m} \left[ \frac{\epsilon_0^2}{64 R^2 \alpha^2 d} \right]
    \label{eq:a1-eq}
\end{equation}

Similarly for \eqref{eq:b-sqrta} it is sufficient that
\[
    \frac{\epsilon_0}{4 \sqrt m} \geq \frac{B}{\sqrt A_0}
\]
Now
\[
    \frac{B}{\sqrt A_0} \geq \frac{2 \sqrt h M \sqrt{\E \norm{\hat \theta - \bar \theta}^2 + d / m}}{\sqrt{mn}}
\]
Leading to the following sufficient condition on $n$
\[
    n \geq \frac{64 h M^2}{\epsilon_0^2} \left[ \E \norm{\hat \theta - \bar \theta}^2 + \frac d m \right].
\]
Now due to the conditions on $h$, define
\[
    \beta := \max \left\{ \frac{1}{2L^2 + 1}, \frac{\epsilon_0^2}{64 L^2 \alpha^2 d} \right\}.
\]
Then \eqref{eq:b-sqrta} will hold when
\begin{equation}
    n \geq \frac{64 L^2 \beta}{\epsilon_0^2} m \left[ \E \norm{\hat \theta - \bar \theta}^2 + \frac d m \right]
\end{equation}
\qed
\end{proof}

\subsection{Proof of Lemma \ref{lem:grad-noise}}

\begin{proof}
    Our proof follows similarly to \cite{Dubey2016},
    \begin{align*}
    \E\norm{\xi_k}^2 
    &= \E\norm{ \nabla \tilde f(\theta_k) - \nabla f(\theta_k) }^2 \\
    &= \E\norm{ \nabla f_0(\theta_k) - \nabla f_0(\hat \theta) + \frac{1}{n} \sum_{i \in S_k} \frac{1}{p_i} \left[ \nabla f_i(\theta_k) - \nabla f_i(\hat \theta) \right] - \left[ \nabla f(\theta_k) - \nabla f(\hat \theta) \right] }^2 \\
    &\leq \frac{1}{n^2} \E \sum_{i \in S_k} \norm{ \left[ \nabla f(\theta_k) - \nabla f(\hat \theta) \right] - \left( \nabla f_0(\theta_k) - \nabla f_0(\hat \theta) + \frac{1}{p_i} \left[ \nabla f_i(\theta_k) - \nabla f_i(\hat \theta) \right] \right) }^2.
\end{align*}
Where the third line follows due to independence. For any random variable $R$, we have that $\E \norm{R - \E R}^2 \leq \E\norm{R}^2$. Using this, the Lipschitz results of Assumption \ref{ass:Lipschitz} and our choice of $p_i$, gives the following, where $\E_I$ refers to expectation with respect to the sampled datum index, $I$, 
\begin{align*}
    \E\norm{\xi_k}^2 
    &\leq \frac{1}{n} \E_I\left( \frac{1}{p_I} \norm{\nabla f_I(\theta_k) - \nabla f_I(\hat \theta)}^2\right) \\
    &\leq \frac{1}{n} \sum_{i=1}^N \frac{\sum_{j=1}^N L_j}{L_i} \left( L_i\norm{\theta_k-\hat\theta}\right)^2 \\
    &= \frac{1}{n} \left\{\sum_{i=1}^N \left(\sum_{j=1}^N L_j\right) {L_i} \right\} \norm{\theta_k-\hat\theta}^2,
\end{align*}
from which the required bound follows trivially.
\qed
\end{proof}

\subsection{Proof of Lemma \ref{lem:lin-just}}

\begin{proof}
\emph{Lipschitz condition:} By the triangle inequality
\begin{align*}
    \norm{\nabla f(\theta) - \nabla f(\theta')}
    &\leq \sum_{i=0}^N \norm{\nabla f_i(\theta) - \sum_{i=0}^N f_i(\theta')} \\
    &\leq (N+1)L \norm{\theta - \theta'}.
\end{align*}

\emph{Strong convexity:} We have that
\begin{align*}
    f(\theta) - f(\theta') - \nabla f(\theta')^\top (\theta - \theta') 
        &= \sum_{i=0}^N \left[ f_i(\theta) - f_i(\theta') - \nabla f_i(\theta')^\top (\theta - \theta') \right] \\
        &\geq \frac{(N+1)l}{2} \norm{\theta - \theta'}_2^2.
\end{align*}
\qed
\end{proof}

\section{Postprocessing proofs}

\subsection{Proof of Theorem \ref{thm:var-reduc}}

\begin{proof}
We start from the bound in Theorem $6.1$ of \citet{Mira2013}, stating for some control variate $h$, the optimal variance reduction $R$ is given by
\[
    R = \frac{ \left( \E_{\theta | \mb x} \left[ g(\theta) h(\theta) \right] \right)^2 }{ \E_{ \theta | \mb x} \left[ h(\theta) \right]^2 },
\]
so that in our case we have
\begin{align*}
    \hat R &= \frac{ \left( \E_{\theta | \mb x} \left[ g(\theta) \hat h(\theta) \right] \right)^2 }{ \E_{ \theta | \mb x} \left[ \hat h(\theta) \right]^2 } \\
    &= \frac{ \left( \E_{\theta | \mb x} \left[ g(\theta) h(\theta) \right] \right)^2 }{ \E_{ \theta | \mb x} \left[ h(\theta) \right]^2 + \frac 1 4 \E_{\theta | \mb x} \left[ \mb a \cdot \xi_S(\theta) \right]^2 } \\
    &= \frac{R}{1 + \frac{ \frac 1 4 \E_{\theta | \mb x} \left[ \mb a \cdot \xi_S(\theta) \right]^2 }{ \E_{\theta | \mb x} \left[ h(\theta) \right]^2 } }.
\end{align*}
Then we can apply Lemmas \ref{lem:top}, \ref{lem:bottom}, defined in Section \ref{sec:zv-lemmas}, to get the desired result
\begin{equation}
    \hat R \geq \frac{ R }{ 1 +  [\sigma(N + 1)]^{-1} \E_{\theta | \mb x} \E_S \norm{ \xi_S(\theta) }^2 }.
\end{equation}
\qed
\end{proof}

\subsection{Lemmas}
\label{sec:zv-lemmas}

\begin{lem}
    Define $A = \sum_{i=1}^d a_i^2$, and let $\xi_S(\theta) = \widehat{ \nabla \log p(\theta|\mb x) } - \nabla \log p(\theta|\mb x)$ be the noise in the gradient estimate. Then
    \[
        \E_{\theta | \mb x} \left[ \mb a \cdot \xi_S(\theta) \right]^2 \leq A \E_{\theta | \mb x} \E_S \norm{\xi_S(\theta)}^2.
    \]
    \label{lem:top}
\end{lem}
\begin{proof}
    We can condition on the gradient noise, and then immediately apply the Cauchy-Schwarz inequality to get
    \begin{align*}
        \E_{\theta | \mb x} \left[ \mb a \cdot \xi_S(\theta) \right]^2 
        &= \E_{\theta | \mb x} \E_S \left[ \mb a \cdot \xi_S(\theta) \right]^2 \\
        & \leq \left( \sum_{i=1}^d a_i^2 \right)
        \E_{\theta | \mb x} \E_S \norm{ \xi_S(\theta) }^2
    \end{align*}
\qed
\end{proof}

\begin{lem}
    Under Assumption \ref{ass:bounded}, define $A = \sum_{i=1}^d a_i^2$. Then $\E_{\theta | x} \left[ h(\theta) \right]^2 \leq A \sigma (N + 1) / 4$.
    \label{lem:bottom}
\end{lem}
\begin{proof}
Applying the Cauchy-Schwarz inequality
    \begin{align*}
        \E_{\theta | x} \left[ h(\theta) \right]^2 & \leq \frac{1}{4} \left( \sum_{i=1}^d a_i^2 \right) \E_{\theta | \mb x} \norm{ \nabla f( \theta ) }^2 \\
        &\leq \frac{A (N + 1)}{4} \sigma
    \end{align*}
\qed
\end{proof}

\end{document}